\documentclass[letterpaper,11pt]{article}

\usepackage[normalem]{ulem}

\usepackage{ifthen}

\usepackage{times}
\usepackage{fullpage}
\usepackage{amsmath}
\usepackage{amssymb}
\usepackage{amsthm} 
\usepackage{bbm}

\usepackage{tabularx}

\usepackage{graphicx} 

\usepackage{ifpdf}
\ifpdf
\usepackage[
  pdftex,
  bookmarks,
  pagebackref,
  plainpages=false, 
  pdfpagelabels=true 
]{hyperref}
\else \fi

\newtheorem{theorem}{Theorem}
\newtheorem{lemma}{Lemma}

\newtheorem{proposition}{Proposition} 
\newtheorem{corollary}{Corollary}

\newtheorem{definition}{Definition}

\newtheorem{remark}{Remark}

\title{\bf Optimal bounds for parity-oblivious random access codes \quad \\ \quad \\ } 

\date{March 23, 2016}

\author{ 
Andr\'e Chailloux\thanks{
INRIA, Paris Rocquencourt, SECRET Project Team. Email: {\tt andre.chailloux@inria.fr}.}
\and
Iordanis Kerenidis\thanks{
IRIF, Universit\'e Paris Diderot, Paris, France and  
Centre for Quantum Technologies, National University of Singapore, Singapore. 
Email: {\tt jkeren@liafa.univ-paris-diderot.fr}.}
\and
Srijita Kundu\thanks{Chennai Mathematical Institute, Chennai, India. Email: {\tt srijita@cmi.ac.in}.}
\and
Jamie Sikora\thanks{
{Centre for Quantum Technologies, National University of Singapore, and MajuLab, CNRS-UNS-NUS-NTU International Joint Research Unit, UMI 3654, Singapore. 
Email: {\tt cqtjwjs@nus.edu.sg}.}}
}

\newcommand{\E}{\mathbb{E}}
\newcommand{\spa}[1]{\mathcal{#1}}
\newcommand{\Tr}{\mathrm{Tr}}
\newcommand{\half}{\frac{1}{2}}

\newcommand{\calA}{\mathcal{A}}
\newcommand{\calB}{\mathcal{B}}

\newcommand{\tr}{\mathrm{Tr}}
\newcommand{\ket}[1]{| #1 \rangle}
\newcommand{\bra}[1]{\langle #1 |}

\newcommand{\ketbra}[2]{\ket{#1} \bra{#2}}
\newcommand{\kb}[1]{\ketbra{#1}{#1}}

\newcommand{\inner}[2]{\langle #1, #2 \rangle}

\newcommand{\set}[1]{\left\{ #1 \right\}}

\newcommand{\Diag}{\mathrm{Diag}}

\newcommand{\diag}{\mathrm{diag}}

\newcommand{\zo}{\{0,1\}}

\newcommand{\RAC}{\textup{RAC}}

\newcommand{\POM}{\textup{PO} \text{-} {\RAC}}
\newcommand{\EPOM}{\textup{EPO} \text{-} \RAC}

\newcommand{\POMn}{\POM^n}

\newcommand{\EPOMn}{\EPOM^n}

\newcommand{\Ind}{\textup{INDEX}}
\newcommand{\Indn}{\textup{INDEX}^n}

\begin{document}

\maketitle

\begin{abstract}
{Random access coding is an information task that has been extensively studied and found many applications in quantum information. In this scenario, Alice receives an $n$-bit string $x$, and wishes to encode $x$ into a quantum state $\rho_x$, such that Bob, when receiving the state $\rho_x$, can choose any bit $i \in [n]$ and recover the input bit $x_i$ with high probability. Here we study two variants: parity-oblivious random access codes, where we impose the cryptographic property that Bob cannot infer any information about the parity of any subset of bits of the input apart from the single bits $x_i$; and even-parity-oblivious random access codes, where Bob cannot infer any information about the parity of any even-size subset of bits of the input. }
 
{In this paper, we provide the optimal bounds for parity-oblivious quantum random access codes and show that they are asymptotically better than the optimal classical ones. Our results provide a large non-contextuality inequality violation and resolve the main open problem in a work of Spekkens, Buzacott, Keehn, Toner, and Pryde (2009). Second, we provide the optimal bounds for even-parity-oblivious random access codes by proving their equivalence to a non-local game and by providing tight bounds for the success probability of the non-local game via semidefinite programming. In the case of even-parity-oblivious random access codes, the cryptographic property holds also in the device-independent model.}
\end{abstract}

\newpage 
\section{Introduction}  
 
Quantum information theory studies how information is encoded in quantum mechanical systems and how it can be transmitted through quantum channels. A main question is whether quantum information is more powerful than classical information. A celebrated result by Holevo \cite{Hol73} shows that quantum information cannot be used to compress classical information. In high level, in order to transmit $n$ uniformly random classical bits, one needs to transmit no less than $n$ quantum bits. This might imply that quantum information is no more powerful than classical information. This however is wrong in many situations. In the model of communication complexity, one can show that transmitting quantum information may result in exponential savings on the communication needed to solve specific problems (\cite{raz99,BCWdW01,bjk:q1way,GavinskyKKRW08,RegevK11}).

One specific information task that has been extensively studied in quantum information is the notion of \emph{random access codes} ($\RAC$s) \cite{Nay99, ANTV99, ANTV02}. 
In this scenario, Alice receives an $n$-bit string $x$, drawn from the uniform distribution, and wishes to encode $x$ into a quantum state $\rho_x$, such that Bob, when receiving the state $\rho_x$, can choose any bit $i \in [n]$ and recover the input bit $x_i$ with high probability by performing some general quantum operation on $\rho_x$. 

$\RAC$s have been used in various situations in quantum information and computation, including in communication complexity, non-locality, extractors and {device-independent cryptography} 
\cite{Ben-AroyaRW08,IwamaNRY07,PZ10,DeV10,PhysRevA.85.052308}.
Even though this task seems easier than transmitting the entire input string $x$, it is known that the length of quantum $\RAC$s must be at least $\Omega(n)$~\cite{Nay99}. In fact, the length of a classical $\RAC$ can be within a logarithmic additive factor of a quantum $\RAC$~\cite{ANTV99}. 

On the other hand, a well-known example shows the advantages of quantum $\RAC$s by using a single qubit to encode two uniformly random classical bits. In this case, the success of correctly decoding either bit is $\cos^2 (\pi/8)$ \cite{BBBW83, ANTV99} while the optimal classical encoding can achieve an average success probability of $3/4$. An advantage can also be proven for the case of encoding three classical bits into one qubit as shown by Chuang (see~\cite{ANTV02} for details), but not for $n \geq 4$~\cite{HINRY}. 

Nevertheless, a question remained of whether there are variants of $\RAC$s, for which we can have an asymptotically significant advantage in the quantum case. We show that this is indeed the case for the so-called \emph{parity-oblivious} $\RAC$s (denoted here as $\POM$s). These are the usual $\RAC$s with the extra cryptographic property that the receiver cannot infer any information about the parity of any subset of bits of the input, apart from the single bits. 

This cryptographic property means, in particular, that once some information about a bit is learned, then no other information can be extracted about any of the other bits. Such a notion has applications in various areas of cryptography. For example, this is a requirement for a class of classical or quantum protocols known as \emph{symmetric-private information retrieval schemes} (PIR)~\cite{GIKM98, KdW04} where one or more servers have a database $x$, a user chooses an index $i$ and at the end, the user learns $x_i$ but no other bit of $x$, and $i$ remains hidden. A parity-oblivious $\RAC$ satisfies the security conditions of a PIR scheme since the index $i$ remains hidden (the $\RAC$ is non-interactive) and the user cannot learn more than one bit of the database.    

Random access codes that are parity-oblivious have been considered before. For example, the previously mentioned $\RAC$s for encoding two or three classical bits in one qubit have this property. It is not hard to check that for any subset of the inputs of size $2$ or greater, Bob's reduced density matrix is exactly the same for the cases where the parity is $0$ or $1$.  In other words, Bob has no information about the parity. These $\RAC$s violate a \emph{non-contextuality inequality} developed by Spekkens, Buzacott, Keehn, Toner, and Pryde \cite{SBKTP09}. This inequality is discussed further in Subsection~\ref{ssect:NC}. 

{
We will also define a weaker variant called \emph{even-parity-oblivious} $\RAC$s (denoted as $\EPOM$s), where the receiver can infer no information about the parity of any even-size subset of the input. These codes are interesting for two reasons: first, they will let us prove tight upper bounds for $\POM$s; and second, due to their equivalence with a non-local game, their cryptographic property holds in the device independent setting. 
}
 
\subsection{Our results} 

We split our results into four sections. We first present the optimal bounds for parity-oblivious quantum $\RAC$s and even-parity-oblivious $\RAC$s. We then contrast this to the classical case (Subsection~\ref{ssect:classical}), discuss a violation of a {\em non-contextuality inequality} (Subsection~\ref{ssect:NC}), then discuss the security of our optimal quantum $\RAC$s in the device-independent model (Subsection~\ref{ssect:DI2}). 
 
\subsubsection{Quantum random access codes and cryptographic security definitions} 
  
Formally, a quantum $\RAC$ of $n$ classical bits is simply a set of quantum states $\{ \rho_x : x \in \zo^n \}$. We suppose Alice chooses $x \in \zo^n$ uniformly at random, prepares the state $\rho_x$, and sends $\rho_x$ to Bob who has a POVM $\{ M^t_0, M^t_1 \}$ where {the subscript $b$ of $M^t_b$} serves as his guess for the $t$-th bit of $x$, denoted $x_t$, for each index ${t \in [n] := \{ 1, \ldots, n \}}$. Suppose that $x_t$ can be decoded with success probability $\frac{1}{2}(1+\alpha_t)$. Then we say that the \emph{bias}, or \emph{worst-case bias}, of the $\RAC$ is 
\[ \min_{t \in [n]} \alpha_t \]  
and the \emph{average-case bias} is 
\[ \mathop{\E}_{t \sim \mu([n])} \alpha_t, \] 
where $\mu$ is the uniform probability distribution.

We consider two cryptographic variants of quantum $\RAC$s in this paper. We are concerned with designing quantum $\RAC$s which \emph{hide} some information about the encoded string $x$ from a potentially cheating Bob. By information being {hidden}, we mean that there exists no measurement which yields a correct guess with probability greater than that of randomly guessing. In particular, we consider the case where Bob cannot learn the value of 
\[ x_S := \bigoplus_{i \in S} x_i, \]  
for certain choices of subset $S \subseteq [n]$,
of Alice's encoded string $x$. We call the value $x_S$ the \emph{$S$-parity of the string $x$.}
  
\begin{definition}[Parity-oblivious and even-parity-oblivious quantum $\RAC$s] 
We say that a quantum $\RAC$ $\{ \rho_x : x \in \zo^n \}$ is \emph{parity-oblivious}, denoted $\POMn$, if the receiver can infer no information about $x_S$ for any subset $S \subseteq [n]$ of size $2$ or greater, when $x$ is chosen uniformly at random. In other words, for all $S \subseteq [n]$ of size $2$ or greater, we have 
\[ 
\frac{1}{2^{n-1}} \sum_{x \, : \, x_S = 0} \rho_x 
= 
\frac{1}{2^{n-1}} \sum_{x \, : \, x_S = 1} \rho_x. 
\] 
We say that a quantum $\RAC$ is \emph{even-parity-oblivious}, denoted $\EPOMn$, if the receiver can infer no information about $x_S$ for any subset $S \subseteq [n]$ of \emph{even} size, $2$ or greater. 
\end{definition} 

Note that the usual treatment of $\RAC$s is to analyze the relationships between the number of encoded bits $n$, the bias $\alpha$, and the encoding dimension of $\rho_x$. Here, we are not concerned with the encoding dimension, but rather the ability to achieve cryptographic security in terms of parity-obliviousness. 
 
In this paper, we present the optimal bias for a quantum $\POMn$ and show that they perform asymptotically better than the optimal classical version. 

\begin{theorem}[Optimal quantum parity-oblivious random access codes] \label{main}
For any {integer $n \geq 2$}, a quantum parity-oblivious random access code of $n$ bits has worst-case bias at most ${1}/{\sqrt{n}}$. Moreover, this bound can be achieved using $\lfloor n/2 \rfloor$ qubits. 
\end{theorem} 

This is in contrast to the classical setting where the optimal  \emph{average-case} bias is provably $1/n$~\cite{SBKTP09} (discussed further in Subsection~\ref{ssect:classical}). 
 
The main idea of the proof of the upper bound is that quantum encodings can be studied through their close relationship to non-local games. Such connections
were 
noted in \cite{OW10} and in \cite{CKS14} it was shown that certain non-local  
games are equivalent to quantum 
encodings in the sense that the optimal average decoding probability is equal to the success probability of the non-local game.  
 
In a non-local game, two non-communicating parties, Alice and Bob, receive some inputs $s$ and $t$, respectively, according to some probability distribution known to Alice and Bob, and must output $a$ and $b$, respectively, such that $(s,t,a,b)$ satisfy some specific condition. For example, in the CHSH game, the condition is $a \oplus b = s \cdot t$. The goal is to find the optimal quantum (resp. classical) success probability of satisfying the condition when Alice and Bob are allowed to share some initial quantum state (resp. shared randomness). 

We now define a very natural non-local game called the $\Ind$ game which {we} use in the analysis in this paper. 

\begin{definition}[$\Ind$ game] 
The \emph{$\Indn$ game}, parameterized by $n$, is the following non-local game:  
\begin{itemize} 
\item Alice's input: Alice receives a random $s$ from the set $S := \zo^{n}$. 
\item Bob's input: Bob receives a random index $t$ from the set $T := [n]$. 
\item Winning condition: They win if Alice's output bit $a$ and Bob's output bit $b$ satisfy $a \oplus b = s_t$. 
\end{itemize} 
\end{definition} 
The choice of initial resource state and local measurement operators (that depend on the respective inputs) comprise a \emph{strategy}. 
We say that a strategy for the $\Indn$ game has \emph{bias} $\alpha$ if 
\[ \mathop{\E}_{s \sim \mu(\zo^n)} \mathop{\E}_{t \sim \mu([n])} \Pr[\textup{Alice's output } a \textup{ and Bob's output } b \textup{ satisfy } a \oplus b = s_t] = \half (1 + \alpha). \] 

We 
show that even-parity-oblivious $\RAC$s with \emph{average-case} bias are equivalent to the $\Ind$ game. In other words, any $\Ind$ game strategy with bias $\alpha$ yields  an even-parity-oblivious $\RAC$ with \emph{average-case} bias $\alpha$ and vice versa.

\begin{theorem}[Equivalence]  
\label{thm:Equiv}
For any $n \in \mathbb{N}$, there exists a quantum even-parity-oblivious $\RAC$ of $n$ 
bits with \emph{average-case} bias $\alpha$ if and only if there exists a quantum $\Indn$ strategy with bias $\alpha$. 
\end{theorem} 
 
Noting that the $\Ind$ game is an XOR game, i.e., the winning condition depends only on the XOR of Alice and Bob's one-bit answers, we use a tight semidefinite programming characterization~\cite{CSUU08} to provide the exact optimal quantum bias. 

\begin{theorem}[Optimal quantum $\Ind$ game bias] \label{thm:QGameBias}
For any $n \in \mathbb{N}$, the optimal quantum bias of an $\Indn$ strategy is $1/\sqrt{n}$.
\end{theorem} 

The above two theorems imply the optimal bounds for even-parity-oblivious random access codes.

\begin{corollary}[Optimal quantum even-parity-oblivious random access codes] \label{eporac}
For any {integer $n \geq 2$}, a quantum even-parity-oblivious random access code of $n$ 
bits has average-case bias at most ${1}/{\sqrt{n}}$. Moreover, this bound can be achieved using $\lfloor n/2 \rfloor$ qubits. 
\end{corollary} 
 
Since the \emph{worst-case} bias of a quantum $\POM$ is obviously upper bounded by the optimal \emph{average-case} bias of a quantum $\RAC$ hiding \emph{only} the even parities, Theorems \ref{thm:Equiv} and \ref{eporac} show that every $\POMn$ 
has bias at most $1/\sqrt{n}$. 

To prove this upper bound is tight, we give an explicit construction of a quantum $\POM$ of $n$ bits with bias $1/\sqrt{n}$ that uses $\lfloor n/2 \rfloor$ qubits and $1$ classical bit. 
This $\RAC$ is based on the notion of \emph{hyperbits}~\cite{PW12} and a proof of Tsirelson's Theorem~\cite{Tsi87}. We then discuss how to remove the classical bit and make it \emph{device-independent} (see Subsection~\ref{ssect:DI2} for more details about the device-independent model). 
 
We remark that parity-oblivious and even-parity-oblivious quantum $\RAC$s 
both share the same worst-case and average-case bias of $1/\sqrt{n}$. However, the same is not true if we consider \emph{odd-parity-oblivious} $\RAC$s where the parities are hidden for only odd-size subsets (greater or equal to $3$). Consider encoding a six-bit string $(x_1, \ldots, x_6)$ where the first three bits are encoded using Chuang's $\POM$ and similarly for the last three bits. It is a straightforward exercise to verify that this is odd-parity-oblivious and that any bit can be decoded with bias $1/\sqrt{3} > 1/\sqrt{6}$. We leave finding the optimal bounds for odd-parity-oblivious $\RAC$s an open problem. 

\subsubsection{Parity-oblivious classical $\RAC$s}
\label{ssect:classical} 

We also study classical $\RAC$s, defined below, for which both variants of bias and both variants of parity-obliviousness are defined analogously.  

\begin{definition}[Classical $\RAC$s with worst-case and average-case biases] 
A \emph{classical $\RAC$} 
is a set of strings 
$\{e(x,r): x \in \zo^n, r \in \zo^m \}$ where $r$ corresponds to private randomness. After choosing ${x \in \zo^n}$ uniformly at random,  Alice samples $r$ from the private randomness, sends to Bob the string $e(x,r)$, and Bob has a decoding procedure given as function $f_t$, for each $t \in [n]$, for learning the $t$'th bit of $x$.  
\end{definition}

We note that the equivalence stated in Theorem~\ref{thm:Equiv} holds in the classical case as well (remarked in Section~\ref{sec:Equiv}). To find the optimal average-case bias of even-parity-oblivious classical $\RAC$s, we provide the following theorem. 
 
\begin{theorem}[Optimal classical $\Ind$ game bias] \label{thm:CGameBias}
For any $n \in \mathbb{N}$, the optimal classical bias of an $\Indn$ strategy is $\sqrt{\frac{2}{\pi n}}(1 + O(1/n))$.
\end{theorem} 

This theorem, together with the classical version of the equivalence shows that classical $\RAC$s that are even-parity-oblivious have an optimal average-case bias of $\sqrt{\frac{2}{\pi n}}(1 + O(1/n))$. Note that, asymptotically, this value is the same as the quantum value, that is, having a bias of $O(1/\sqrt{n})$. However, differences arise when one considers $\RAC$s that also hide the odd parities. Consider the following proposition of Spekkens, Buzacott, Keehn, Toner, and Pryde. 

\begin{proposition}[Optimal parity-oblivious classical $\RAC$s \textup{\cite{SBKTP09}}] 
For any $n \in \mathbb{N}$, a parity-oblivious classical $\RAC$ of $n$ 
bits has \emph{average-case bias} at most $1/n$. Moreover, this bound can be achieved using $1$ classical bit. 
\end{proposition} 
 
Thus, there is a difference between the optimal average-case biases of parity-oblivious and even-parity-oblivious $\RAC$s in the classical setting, in contrast to the quantum setting. 

\subsubsection{Large non-contextuality inequality violations} \label{ssect:NC}

The basic primitives in an operational theory are preparations and measurements which can be thought of as instructions for the laboratory apparatus. For example, the operational theory can be given in terms of {\em hidden variables} which are probability distributions characterizing the outcomes of the preparations and measurements. That is, a preparation creates a {\em physical state} (each occurring with some probability) and a measurement acts upon a physical state and outputs a {\em prediction} or simply an {\em outcome} (each occurring with some probability). Thus, the probability distributions characterizing these actions are how they are represented in this model.
  
A hidden variable model is {\em preparation 
non-contextual} if whenever two preparations yield the same statistics for all possible measurements then they are represented equivalently in the model and a hidden variable model is {\em measurement 
non-contextual} if whenever two measurements have the same statistics for all preparations then they are represented equivalently in the model (see  \cite{SBKTP09} and references therein for a more thorough discussion). Similar to non-locality, a non-contextuality inequality is any inequality on probability distributions that follows from the assumption that there exists a hidden variable model that is preparation or measurement non-contextual.

Spekkens, Buzacott, Keehn, Toner, and Pryde \cite{SBKTP09} proved the following \emph{non-contextuality inequality} (or NC inequality, for short). 

\begin{proposition}[Non-contextuality inequality \textup{\cite{SBKTP09}}] \label{clPOMn}
In any operational theory that admits a preparation non-contextual hidden variable model, the \emph{average-case bias} for any parity-oblivious $\RAC$ is at most $1/n$. 
\end{proposition}

Then, they discussed that quantum mechanics violates this NC inequality for $n \in \{ 2, 3 \}$, by noting the previously mentioned parity-oblivious quantum $\RAC$s of two and three classical bits into one qubit with respective average-case biases of $\frac{1}{\sqrt{2}}$ and $\frac{1}{\sqrt{3}}$. It was left as an open question whether quantum mechanics violates this NC inequality for $n \geq 4$. 

Through our analysis, we have shown that the optimal average-case bias for quantum parity-oblivious $\RAC$s is $1/\sqrt{n}$, thus resolving their main open question. This provides a family of NC inequality violations that grow with the input size $n$. 

Note, that if there exists a game for which the winning probability of any classical strategy cannot deviate from $1/2$ by more than $\delta_1$ and, moreover, there is
a quantum strategy with winning probability at least $1/2+\delta_2$, then we can obtain a violation of order $\delta_2/\delta_1$ (see \cite{BRSdW12} for details).  
Hence, to quantify the violation of this NC inequality, we consider the ratio of the optimal average-case bias of quantum parity-oblivious $\RAC$s and that of any operational theory  admitting a preparation non-contextual hidden variable model. More precisely, we show an explicit non-contextuality inequality  violation of order $\sqrt{n}$.

\begin{theorem} \label{context}
For any {integer $n \geq 2$}, there exists an explicit non-contextuality inequality that provides a violation of order $\sqrt{n}$. 
\end{theorem}

Note that other large non-contextuality inequality violations have been found, see for example the work of Vidick and Wehner~\cite{VW11}. 

\subsubsection{Device-independent quantum $\RAC$s} \label{ssect:DI2}

Until this point, we have discussed the bias and the parity-obliviousness of a quantum $\RAC$ which are functions of the encoding states $\{ \rho_x : x \in \zo^n \}$ only. However, much of the cryptographic analysis in this paper is concerned with how the states $\rho_x$ are {\em prepared}. In this subsection, we discuss how the security of the $\RAC$ is affected if one cannot trust the quantum apparatus used in the preparation of $\rho_x$. 

The device-independent model of cryptographic security deals with the setting when the devices used in the protocol are not trusted, or are even malicious, being created by the cheaters/eavesdroppers themselves. Many security proofs in this setting are based on quantum non-locality or the no-signalling principle, each having {their} own {limitations} which ultimately limits the cheating capabilities for anyone controlling the preparation and/or execution of the quantum devices in the protocol. {Recall that the no-signalling principle, which is satisfied by the laws of quantum mechanics, roughly states that it is impossible to {\em send information} arbitrarily fast, in particular faster than the speed of light. For example, in a quantum setting, Alice cannot convey information to a distant Bob by simply measuring her half of a shared quantum state.} 
 
Obviously, if the preparation of the encoding $\rho_x$ is as simple as Alice having a quantum device which outputs $\rho_x$ on input $x$, then certainly device-independence is not feasible since Bob may control the quantum device and just have it prepare $\rho_x = \kb{x}$ (or some other function of $x$ according to {what} he wishes to learn). However, the preparation need not be so simple. We now sketch the preparation of the quantum $\RAC$s presented in this work to give an idea of how they can be device-independent. 
 
First, Alice creates a bipartite quantum state $\ket{\psi}$ and sends a subsystem to Bob. {\em Afterwards} she chooses a string $s \in \zo^n$ uniformly at random and measures her half of the state to get an outcome $a \in \zo$. She then defines $x_t := s_t \oplus a$, for all $t \in [n]$, and Bob's post-measured state is now his encoding of the string $x$. 
Since there is no communication from Alice to Bob after Alice chooses $s$, he must not be able to infer any information about $s$ from his encoding of $x$. Thus, Bob has limited information of any function of $x$ which contains information about $s$. For example, Bob cannot learn $x_1 \oplus x_2$ since 
\[ x_1 \oplus x_2 = (s_1 \oplus a) \oplus (s_2 \oplus a) = s_1 \oplus s_2 \] 
which is hidden {by} the no-signalling principle. Therefore, even if Bob created the entire state which Alice shares at the beginning, and Alice's measurement, he cannot infer any information about $x_1 \oplus x_2$, promised only by the no-signalling principle. 

\begin{theorem} \label{thm:DI}
There exists a preparation of an optimal $\POMn$ with bias $1/\sqrt{n}$ which is even-parity-oblivious in the device-independent model against a no-signalling Bob. 
\end{theorem}  

We prove the above theorem using a small modification of our  optimal $\POMn$ in Section~\ref{sec:OptEnc}. 
See Subsection~\ref{ssect:DI} for more details. {Note also that the above theorem implies that there exist optimal even-parity-oblivious random access codes which retain their cryptographic property in the device independent model.}

\subsection{Organization of the paper} \label{ssect:organ}

In Section~\ref{sec:Equiv}, we prove the equivalence of even-parity-oblivious $\RAC$s and $\Indn$ strategies. In Section~\ref{sec:GameBias} we discuss the optimal quantum and classical bias of the $\Indn$ game for any $n$. We conclude in Section~\ref{sec:OptEnc} by presenting an optimal parity-oblivious quantum $\RAC$ and {prove the security for the even-parity-obliviousness} in the device-independent model. 

\section{Equivalence of $\EPOM$s and $\Indn$ strategies} \label{sec:Equiv} 
  
In this section we prove the equivalence in Theorem~\ref{thm:Equiv}, reproduced below. 

\setcounter{theorem}{1}
\begin{theorem}[Equivalence]  
For any $n \in \mathbb{N}$, there exists a quantum even-parity-oblivious $\RAC$ of $n$ uniformly random classical bits with \emph{average-case} bias $\alpha$ if and only if there exists a quantum $\Indn$ strategy with bias $\alpha$. 
\end{theorem} 

{For this reason, even-parity-obliviousness of a $\RAC$ is a very natural property. In particular, in the simple reduction from $\Ind$ strategies to $\RAC$s (Subsection~\ref{ssec:GameToProtocol}), we see how even-parity-obliviousness appears and how the $\RAC$ may not hide the odd parities.} 

\subsection{From $\RAC$s to $\Ind$ strategies} 

Let us fix an even-parity-oblivious $\RAC$ $\{ \rho_x : x \in \zo^n \}$ with average-case bias $\alpha$. Let $\spa{B}$ be the Hilbert space used for the encoding. 
Our goal is to construct a strategy for $\Indn$ with bias $\alpha$. For each $\rho_x $, we fix a purification $\ket{\psi_x}$ of $\rho_x$ in the space $\spa{A} \otimes \spa{B}$. For $a\in \{0,1\}$, let $\boldsymbol{a}$ be the $n$-bit string $(a,\dots,a)$ and $\bar{s}$ be the bit-wise complement of a string $s$. {For $s \in \{ 0, 1 \}^n$}, define the following state 
\[
\ket{\Omega_s} := \frac{1}{\sqrt{2}} \sum_{a \in \zo} \ket{a}_\spa{O}\ket{\psi_{s \oplus \boldsymbol{a} }}_\spa{AB} = \frac{1}{\sqrt{2}} \ket{0}\ket{\psi_{s}} + \frac{1}{\sqrt{2}} \ket{1}\ket{\psi_{\bar{s}}},   
\]
where $\spa{O}$ is a qubit register containing the value of $a$. 
We would like to show that if Bob has the register $\spa{B}$ of the above state, then he has no information about $s$. Note that his reduced state is $\sigma_s := \frac{1}{2} \rho_s + \frac{1}{2} \rho_{\bar{s}}$.

The first step is to see that Bob has no information about any parity of $s$ (not even of the values of the singleton bits). Fix an arbitrary, non-empty subset $S$. 
{For fixed $b \in \{ 0, 1 \}$, Bob's reduced state, averaged over all $s \in \{ 0, 1\}^n$ such that $s_S=b$, is given by
\[ \sigma^b_S := \frac{1}{2^{n-1}} \sum_{s: s_S=b} \sigma_s = \frac{1}{2^n}  \left( \sum_{s:s_S=b} \rho_s + \sum_{s:s_{S}=b} \rho_{\bar s} \right). \]
Note that $s_S=\bar{s}_S$ when $|S|$ is even and $s_S=\bar{s}_{S} \oplus 1$ when $|S|$ is odd. Thus, by  defining $\rho^b_S$ in the similar way 
\[ \rho_S^b := \frac{1}{2^{n-1}} \sum_{s: s_S=b} \rho_s \]
we can easily verify that $\sigma_S^b = \rho_S^b$ when $|S|$ is even and $\sigma_S^b = \half \rho_S^0 + \half \rho_S^1$ when $|S|$ is odd. Note that since $\{ \rho_x : x \in \{ 0, 1 \}^n \}$ is an even-parity-oblivious $\RAC$, we have by definition that $\rho_S^0 = \rho_S^1$ for $|S|$ even (otherwise, Bob could measure to learn some information about the even parity). Thus, we have that $\sigma_S^0 = \sigma_S^1$ for all nonempty subsets $S$ and therefore   
all the parities are hidden from Bob when given $\sigma_s$ (when $s$ is chosen uniformly at random).}
This means that for any nonempty subset $S$ and measurement $M$, Bob has a maximum probability of $1/2$ of successfully guessing $s_S$ from the $\RAC$ $\{ \sigma_s : s \in \{ 0, 1 \}^n \}$. 

In the following lemma, we prove that if an encoding reveals no information about the parity of any subset, then the encoding reveals no information about the string. This is intuitively an obvious statement that we rigorously prove below.

\begin{lemma} 
If an encoding $\set{ \sigma_s : s \in \zo^n }$ satisfies $\mathop{\E}_{s \sim \mu(\zo^n)} \Pr[\textup{learn } s_S] = \half$, 
for every subset $S \subseteq [n] \setminus \emptyset$, then $\sigma_s = \sigma_{s'}$ for all $s, s' \in \zo^n$. 
\end{lemma} 

\begin{proof}
Suppose for a contradiction that there exists $s, s' \in \zo^n$ such that $\sigma_s \neq \sigma_{s'}$. Then there exists a subset $T {\subseteq} \zo^{n}$ of size $2^{n-1}$ such that $\sigma_T = \frac{1}{2^{n-1}} \sum_{s \in T} \sigma_s $ is not equal to $\sigma_{\bar{T}} = \frac{1}{2^{n-1}} \sum_{s \in \bar{T}} \sigma_s$, {where ${\bar T}$ denotes the complement of the set $T$}. To see this, take any subset $T {\subseteq} \zo^{n}$ of size $2^{n-1}$; if $\sigma_T = \sigma_{\bar{T}}$, then we can find $s \in T$ and $s' \in \bar{T}$ such that $\sigma_s \neq \sigma_{s'}$, since all the $\sigma_i$ are not equal. We consider the subset $T'$ where we add $\{s'\}$ and remove $\{s\}$ from $T$ to obtain $\sigma_{T'} \neq \sigma_{\bar{T'}}$. 
 
This means that there exists a two-outcome measurement {$\{ M_T, M_{\bar T} \}$} that outputs 1 if $s \in T$ and $-1$ otherwise, with positive bias. We now show for a contradiction that this measurement must also output a parity of some nonempty subset with positive bias. 
Define the function $f:\{0,1\}^n \rightarrow \{-1,+1\}$ as the indicator function of $T$ {and let $b$ be the expectation over the measurement outcomes when measuring $\sigma_s$ with $\{ M_T, M_{\bar T} \}$, so $b(s) := \tr(\sigma_s M_{T}) - \tr(\sigma_s M_{\bar{T}})$}. Then 
\[
\mathop{\E}_{s \sim \mu(\zo^n)} [{b(s)} \cdot f(s)] > 0. 
\]
By taking the Fourier representation of the function, we have 
\[ \mathop{\E}_{s \sim \mu(\zo^n)} [{b(s)} \cdot f(s)] 
= 
\mathop{\E}_{s \sim \mu(\zo^n)} \!\! \left[ {b(s)} \cdot \sum_{S \subseteq [n]} \hat{f}(S) \, (-1)^{s_S} \right] \!\! 
= 
\sum_{S \subseteq [n]} \hat{f}(S) \mathop{\E}_{s \sim \mu(\zo^n)} \left[ {b(s)} \cdot (-1)^{s_S} \right] 
>  
0. \] 
Note that $\hat{f}(\emptyset) = \mathop{\E}[f(s)] = 0$,  {because $|T| = |\bar{T}|$}, 
implying that there exists a non-empty subset $S$ for which 
\[ \mathop{\E}_{s \sim \mu(\zo^n)} [{b(s)} \cdot (-1)^{s_S}] \neq 0, \] 
which is a contradiction. 
\end{proof}

The above statement means that for each $s$, we have $\Tr_{\spa{OA}} \kb{\Omega_s} = \Tr_{\spa{OA}}\kb{\Omega_0}$. In particular, for any $s \in \zo^n$ there exists a unitary $U_s$ acting on $\spa{OA}$ such that $(U_s \otimes I)\ket{\Omega_0} = \ket{\Omega_s}$. We use the state $\ket{\Omega_0}$ to define the $\Indn$ strategy:

\begin{itemize}
\item Alice and Bob share the state $\ket{\Omega_0} \in \calA \otimes \calB$.
\item Upon receiving $s \in \zo^{n}$, Alice applies $U_s$ on $\spa{OA}$ such that Alice and Bob share $\ket{\Omega_s}$. Alice measures register $\spa{O}$ in the computational basis and outputs the measurement outcome $a$. 
\item For Alice's input $s$ and output $a$, Bob has an encoding $\rho_{x}$ where $x := s \oplus \boldsymbol{a}$ occurs uniformly at random. Upon receiving $t \in [n]$, Bob measures $\calB$ just as in the $\RAC$ to learn $x_t$.  He outputs $b$ equal to his guess. 
\item Alice and Bob win the game if $b = s_t \oplus a = x_t$ meaning that they win the game if and only if Bob correctly guesses $x_t$. 
\end{itemize}

Since the $\RAC$ has average-case bias $\alpha$, we see that with this $\Indn$ strategy, they succeed with probability
\begin{eqnarray*} 
& & \mathop{\E}_{s \sim \mu(\zo^n)} \mathop{\E}_{t \sim \mu([n])}
\Pr[\textup{Alice's output } a \textup{ and Bob's output } b \textup{ satisfy } a \oplus b = s_t] \\
& = & 
\mathop{\E}_{x \sim \mu(\zo^n)} \mathop{\E}_{t \sim \mu([n])}
\Pr[\textup{Bob correctly outputs } x_t \textup{ from the } \set{\rho_x : x \in \zo^n} \RAC] \\
& = & \half (1 + \alpha), 
\end{eqnarray*} 
as desired. 
 
\subsection{From $\Ind$ strategies to $\RAC$} 
\label{ssec:GameToProtocol} 

Suppose Alice and Bob have a strategy to win the $\Indn$ game with bias $\alpha$ with starting state ${\ket{\psi} \in \calA \otimes \calB}$. On input $s \in \zo^n$, Alice performs on her side the corresponding measurement which generates her outcome $a$. We assume that $a$ is uniformly random and independent of $s$ (which can be guaranteed by taking the XOR with an independently and uniformly random bit that is shared with Bob). 
Let $\rho_{s,a}$ be the state that Bob has when Alice has input  $s$ and outputs $a$
{and define the $\RAC$ $\{ \sigma_x : x \in \{ 0, 1 \}^n \}$ where $\sigma_x := \half \rho_{x,0} + \half \rho_{\bar{x}, 1}$ for each $x \in \{ 0, 1 \}^n$. 
}
We now show that $\{ \sigma_x : x \in \zo^n \}$ is an even-parity-oblivious $\RAC$ with average-case bias $\alpha$. {Note that $\half \rho_{s,0} + \half \rho_{s,1}$ is independent of $s$ by the no-signalling principle. For convenience, define $\rho := \half \rho_{s,0} + \half \rho_{s,1}$ for any $s \in \{ 0, 1 \}^n$.}
\begin{enumerate}
\item It hides the even parities: Let $S \subseteq \{ 0, 1 \}^n$ be a subset of even size and $b \in \{ 0, 1 \}$ be an arbitrary bit. Then we have $x_S=\bar{x}_S$ for any $x \in \{ 0, 1 \}^n$, since $|S|$ is even. Bob's reduced state, averaged over all $x \in \{ 0, 1 \}^n$ such that $x_S = b$, is given by 
\[ 
\frac{1}{2^{n-1}} \sum_{x : x_S = b} \sigma_x 
=
\frac{1}{2^{n}} \sum_{x : x_S = b} \rho_{x,0} 
+ 
\frac{1}{2^{n}} \sum_{x : x_S = b} \rho_{\bar{x},1}  \\  
=
\frac{1}{2^{n}} \sum_{x : x_S = b} \rho_{x,0} 
+ 
\frac{1}{2^{n}} \sum_{x : x_S = b} \rho_{x,1}  \\
=
\rho, 
\] 
which is independent of $b$, thus proving the $\RAC$   is even-parity-oblivious. 
 
\item Since Alice and Bob win the $\Indn$ game with  {average-case bias} $\alpha$, {we know that
\[ 
\half(1 + \alpha) 
= 
\mathop{\E}_{a \sim \mu(\zo)} \mathop{\E}_{s \sim \mu(\zo^n)} \mathop{\E}_{t \sim \mu([n])} \Pr[\text{Bob learns } s_t \oplus a \text{ from } \rho_{s,a}]. \] 
By defining $x := s \oplus \boldsymbol{a}$, we can write the above as}  
\begin{eqnarray*} 
\half(1 + \alpha) 
& = & 
\mathop{\E}_{a \sim \mu(\zo)} \mathop{\E}_{x \sim \mu(\zo^n)} \mathop{\E}_{t \sim \mu([n])} \Pr[\text{Bob learns } x_t \text{ from } \rho_{x \oplus \boldsymbol{a}, a}] \\ 
& = &
\mathop{\E}_{x \sim \mu(\zo^n)} \mathop{\E}_{t \sim \mu([n])} \Pr[\text{Bob learns } x_t \text{ from } \sigma_x] 
\end{eqnarray*} 
{as desired.} 
\end{enumerate} 

{Note that in the proof above, we are treating $x$, the string Alice wishes to encode, as $s \oplus \boldsymbol{a}$. We now remark that some of the odd parities of $x$ may not be hidden from Bob. For example, if in the $\Ind$ game Alice simply outputs $a = s_1 \oplus s_2 \oplus s_3 \oplus d$, where $d$ is the uniformly random bit Alice and Bob share to make $a$ independent of $s$, then we have $x_1 \oplus x_2 \oplus x_3 = d$ and Bob would know this odd parity exactly. However, the even parities of $x$ are equal to those of $s$ which are hidden by the no-signalling principle.} 

\begin{remark} 
The above equivalence also holds in the classical setting.  
\end{remark} 
 
\section{On the structure of optimal $\Ind$ game strategies} \label{sec:GameBias} 

In this section, we prove Theorems~\ref{thm:QGameBias} and \ref{thm:CGameBias}, that the optimal quantum bias of an $\Indn$ strategy is $1/\sqrt{n}$ and the optimal classical bias of an $\Indn$ strategy is $\sqrt{\frac{2}{\pi n}}(1 + O(1/n))$.

\subsection{The quantum bias}


The quantum bias of any XOR game can be found efficiently by solving a semidefinite program (SDP) \cite{CSUU08}. The optimization takes place over a matrix indexed by $s \in S$ and $t \in T$ with each entry corresponding to the expectation of the measurement outcome of a fixed game strategy. Such a matrix of inner products can be written as a positive semidefinite matrix and the expectation (or bias) of the game strategy is then an inner product of this matrix and one containing the information of the XOR game. 

Specifically, the quantum bias of the $\Indn$ game can be calculated as the optimal value of either SDP below 

\vspace{0.25cm}
\begin{center}
  \begin{minipage}{2in}
    \centerline{\underline{Primal problem (P)}}\vspace{-7mm}
    \begin{align*}
	     \text{supremum:} \quad & \inner{B}{X} \\
  		\text{subject to:} \quad & \diag(X) = e \\
  		& X \succeq 0 
		\end{align*}
  \end{minipage}
  \hspace*{25mm}
  \begin{minipage}{2in}
    \centerline{\underline{Dual problem (D)}}\vspace{-7mm}
		\begin{align*}
			\text{infimum:} \quad & \inner{e}{y} \\
  		\text{subject to:} \quad & \Diag(y) \succeq B \\
  	\end{align*}	
  \end{minipage}
\end{center}
where 
\begin{itemize}
\item $\diag(X)$ is the vector on the diagonal of the square matrix $X$, 
\item $e$ is the vector of all ones, 
\item $\Diag(y)$ is the diagonal matrix with the vector $y$ on the diagonal, 
\item $B := \dfrac{1}{2} \left[ \begin{array}{cc} 0 & A \\ A^\top & 0 \end{array} \right]$, \; where $A_{s,t} := \dfrac{(-1)^{s_t}}{n 2^{n}}$. 
\end{itemize} 

For (P), consider the positive semidefinite matrix $X := YY^{\top}$, where 
\[ 
Y := 
\left[ \begin{array}{c} \sqrt{n} \, 2^{n} A 
\\ I_T \end{array} \right] \; . 
\] 

To show $X$ is feasible in (P), one can check that each diagonal entry of $X$ is equal to $1$ from the definition of $A$ above. Note that $\inner{B}{{X}} := \sqrt{n} \, 2^{n} \inner{A}{A} = 1/\sqrt{n}$ proving that the quantum bias is at least $1/\sqrt{n}$ (since the quantum bias is the maximum of $\inner{B}{X}$ over all feasible $X$). 

For (D), let ${y} := \left[ \begin{array}{c} u \, e_S \\ v \, e_T \end{array} \right]$ where $u,v > 0$ (determined later) and $e_S$ and $e_T$ are the vectors of all ones indexed by entries in $S$ and $T$, respectively. Then 
\[ \Diag({y}) \succeq B \iff \left[ \begin{array}{cc} u I_S & - \frac{1}{2} A \\ - \frac{1}{2} A^{\top} & v I_T \end{array} \right] \succeq 0 \iff u v I_T \succeq \frac{1}{4} A^{\top} A = \dfrac{1}{4 n^2 2^{n}} I_T \iff u v \geq \dfrac{1}{4 n^2 2^{n}}. \] 
From above, if we set $v := \dfrac{1}{2 n \sqrt{n}}$ and $u := \dfrac{1}{2 \sqrt{n} \, 2^{n}}$, then ${y}$ is feasible in (D). Since 
\[ {\inner{e}{{y}} = 2^{n} u + nv = \dfrac{1}{\sqrt n}}, \] 
we know the quantum bias is at most $1/\sqrt{n}$ (since the quantum bias is  equal to the minimum of $\inner{e}{y}$ over all feasible $y$). Therefore, the quantum bias is exactly $1/\sqrt{n}$, as required. 
 
The $\Ind$ game turns out to be equivalent to the Retrieval game studied in \cite{OW10} which is defined similarly except the first bit of Alice's input is always $0$ and the other $n-1$ bits are chosen independently and uniformly at random. To see the equivalence, notice that in the $\Ind$ game Alice can take her input $s \in \zo^n$, define $s' = \boldsymbol{m} \oplus s$, where $m$ fixes the specific bit to a specific value, play the Retrieval game strategy with input $s'$ to generate $a'$, and then  output $a := a' \oplus m$ (Bob plays the same strategy). Thus, any strategy for the Retrieval game with bias $\alpha$ yields a strategy for the $\Ind$ game with bias $\alpha$ as well.  
We further remark that the quantum bias of the Retrieval game is shown to be $1/\sqrt{n}$ in \cite{OW10} through the use of uncertainty relations. Using this result, and the equivalence to the $\Ind$ game, we have another proof that the quantum bias of the INDEX game is $1/\sqrt{n}$.  
  
\subsection{The classical bias}\label{Section:ClassicalValue} 
 
We can assume without loss of generality that Alice and Bob's strategies are deterministic. Define ${b \in \{0,1\}^n}$ as the string of potential answers Bob gives where $b_t$ is the bit that Bob outputs on input $t \in [n]$. Now let us examine Alice's strategy. For a fixed input $s$, if she outputs $1$, they win the game with probability 
\[ \mathop{\E}_{t \sim \mu([n])} \Pr[b_t \neq s_t] = \frac{1}{n} |b \oplus s|_H, \] 
where $|x|_H$ denotes the Hamming weight of a string $x \in \zo^n$. If she outputs $0$, they win the game with probability 
\[ \mathop{\E}_{t \sim \mu([n])} \Pr[b_t = s_t] = 1 - \frac{1}{n}|b \oplus s|_H. \] 
Since their strategies are deterministic, Alice should output the maximum of these two, so 
\[ 
\max \left\{ \frac{1}{n} |b \oplus s|_H, 1 - \frac{1}{n} |b \oplus s|_H \right\} 
= \half + \left| \half - \frac{1}{n} |b \oplus s|_H \right|
= \half + \frac{1}{2} \cdot \frac{2}{n} \left| \frac{n}{2} - |b \oplus s|_H \right| . \] 
Therefore, the classical bias is {\em precisely} 
$\frac{2}{n} \mathop{\E}_{s \sim \mu(\zo^n)} \left| \frac{n}{2} - |b \oplus s|_H \right|$. {Note that this quantity is independent of $b$, thus we could assume Bob always outputs $0$ for every input.} The quantity 
\[ \mathop{\E}_{s \sim \mu(\zo^n)} \left| \frac{n}{2} - |b \oplus s|_H \right| \] 
corresponds to the mean deviation of the uniform binomial distribution. This is a well studied quantity  \cite{Fra45} and we know that
\[ \mathop{\E}_{s \sim \mu(\zo^n)} \left[ \left| \frac{n}{2} - \left| b \oplus s \right|_H \right| \right] =  \sqrt{\frac{n}{2\pi}} \left( 1 + O \left( \frac{1}{n} \right) \right). \] 
Therefore, the classical bias is 
$\frac{2}{n} \sqrt{\frac{n}{2\pi}} \left( 1 + O \left( \frac{1}{n} \right) \right) = \sqrt{\frac{2}{\pi n}}(1 + O(\frac{1}{n}))$, as desired. 

\section{A construction of a quantum $\POMn$ with optimal bias} \label{sec:OptEnc}
  
In this section, we give an explicit construction of a quantum $\POMn$ with optimal bias.

\begin{lemma}[Optimal $\POMn$] \label{lem:OptEnc}
For any {integer $n \geq 2$}, there exists a $\POMn$ with bias $1/\sqrt{n}$ that uses $\lfloor n/2 \rfloor$ qubits. 
\end{lemma} 
 
Our construction builds upon the previously mentioned $\RAC$s for sending $2$ (resp. $3$) classical bits with bias $1/\sqrt{2}$ (resp. $1/\sqrt{3}$). These are the vertices from the corners of a square inscribed in an equatorial plane in the Bloch sphere, and the corners of a cube inscribed in the Bloch sphere, respectively. To generalize this idea to an $n$-cube inscribed in an $n$-dimensional sphere, we use the intuition of \emph{hyperbits}, which are a way to visualize such unit vectors in a quantum mechanical setting. A full discussion of hyperbits and their equivalence to certain quantum protocols is beyond the scope of this paper, but we refer the interested reader to the work of Pawlowski and Winter~\cite{PW12}. 

We note that, after the publication of this paper, we became aware that a similar $\RAC$ had been previously discovered by Wehner~\cite{W08}, but remained unpublished.

\subsection{The construction}
 
Our construction is very similar to a proof of Tsirelson's Theorem~\cite{Tsi87}. We start by recursively defining the  observables $G_{n,1}, \ldots, G_{n,n}$, for $n \geq 2$, which are used to define the actions of Alice and Bob in the $\POMn$. 
For $n=2$ and $n=3$, we define  
\[ G_{2,1} := X, \quad G_{2,2} := Y 
\qquad \textup{and} \qquad 
G_{3,1} := X, \quad G_{3,2} := Y, \quad G_{3,3} := Z. \] 
We use the $n=3$ observables as a base case for a recursive formula: 
\begin{align*}
n \textup{ even} : & \quad G_{n,i} := G_{n-1, i} \otimes X,  \; \text{ for } \; i \in \{ 1, \ldots, n-1 \}, & \quad G_{n,n} = I \otimes Y,  \quad & \quad \\  
n \textup{ odd} : & \quad G_{n,i} := G_{n-2, i} \otimes X, \; \text{ for } \; i \in \{ 1, \ldots, n-2 \}, & \quad G_{n,n-1} = I \otimes Y, & \quad G_{n,n} = I \otimes Z. 
\end{align*}

Note that these act on $\lfloor n/2 \rfloor$ qubits,\footnote{We note here that the choice of these observables is not unique and there are applications in the literature that use slightly different observables. However, this particular choice reduces the $\RAC$ dimension by one qubit when $n$ is odd. For example, for $n=3$ our $\RAC$ uses $\lfloor n/2 \rfloor = 1$ qubit (as opposed to {$2$}) just as in the well-known quantum $\RAC$ of three classical bits into one qubit.} have eigenvalues $\pm 1$, and satisfy the anti-commutation relation 
\[  \{ G_{n,i}, G_{n,j} \} = 2 \delta_{i,j} I. \] 
Define the following operators for $x \in \{ 0, 1 \}^n$ and $t \in [n]$: 
\[ A_x := \frac{1}{\sqrt{n}} \sum_{i=1}^n (-1)^{x_i} G_{n,i} \quad \textup{ and } \quad B_t := G_{n,t}^{\top}. \] 
Note that $A_x^2 = I$, for all $x \in \zo^n$, and $B_t^2 = I$, for all $t \in [n]$, so each have $\pm 1$ eigenvalues.

The $\POMn$ protocol is defined below. 
\begin{itemize} 
\item Encoding states: Alice chooses a uniformly random $x \in \zo^n$, creates $\lfloor n/2 \rfloor$ EPR pairs, and measures the first ``halves'' with the observable $A_x$ to get an outcome $a \in \{ -1, +1 \}$. {The second ``halves'' now contain the post-measurement state $\tau_{x,a}$ and since $\tr(A_x) = 0$, each $a$ occurs with $1/2$ probability. She sends $\tau_{x,a}$ and $a$ to Bob who now has the mixed state} 
\[ {\rho_x := \half \sum_a \tau_{x,a} \otimes \kb{a}}\] 
encoding the string $x$.   
\item Decoding procedure: If Bob wishes to learn $x_t$, he measures his EPR halves with the observable $B_t$ to get an outcome $b \in \{ -1, +1 \}$ {and also measures to learn $a$}. He computes $c = ab$ and outputs $0$ if $c = +1$, and $1$ otherwise. 
\end{itemize} 
 
In the next two lemmas, we show that the worst-case bias of this $\RAC$ is $\frac{1}{\sqrt{n}}$ and that it is parity-oblivious, thereby proving Lemma~\ref{lem:OptEnc}.
   
\begin{lemma} 
The quantum $\RAC$ $\{ \rho_x : x \in \{ 0, 1\}^n \}$ has worst-case bias at least $1/\sqrt{n}$. 
\end{lemma} 

\begin{proof} 
We can assume at the beginning of the protocol, Alice and Bob share the maximally entangled state
\[ \ket{\psi} := \frac{1}{\sqrt{2^{\lfloor\frac{n}{2}\rfloor}}}\sum_{j=1}^{2^{\lfloor \frac{n}{2}\rfloor}} \ket{j}_\mathcal{A}\ket{j}_\mathcal{B}.\]
The expectation value of the observable $C = A_x \otimes B_t$ in this state is given by:
\[ 
\langle C \rangle = \bra{\psi} A_x \otimes B_t \ket{\psi} 
= 
\frac{1}{\sqrt{n}} \frac{1}{2^{\lfloor\frac{n}{2}\rfloor}} 
\sum_{i=1}^n (-1)^{x_i} 
\underbrace{\sum_{j,k = 1}^{2^{\lfloor \frac{n}{2}\rfloor}} \bra{j}_\mathcal{A}\bra{j}_\mathcal{B} \, G_{n,i} \otimes G_{n,t}^{\top} \, \ket{k}_\mathcal{A}\ket{k}_\mathcal{B}}_{= 2^{\lfloor\frac{n}{2}\rfloor}  \delta_{i,t}} 
= 
\frac{(-1)^{x_t}}{\sqrt{n}} \] 
where the third equality is derived from the anti-commutation relation. We can write 
\[ \langle C \rangle = \text{Pr}[c = +1] - \text{Pr}[c = -1] = \bra{\psi} A_x \otimes B_t \ket{\psi} \] 
implying 
\[ \Pr[\text{Bob outputs 0}] = \Pr[c = +1] = \frac{1}{2}\left[1 + \frac{(-1)^{x_t}}{\sqrt{n}}\right] \]
\[ \Pr[\text{Bob outputs 1}] = \Pr[c = -1] = \frac{1}{2}\left[1 - \frac{(-1)^{x_t}}{\sqrt{n}}\right]. \]
This proves that 
\[ \Pr[\text{Bob outputs $x_t$}] = \frac{1}{2}\left(1 + \frac{1}{\sqrt{n}}\right), \] 
as desired. 
\end{proof} 

\begin{lemma} 
The quantum $\RAC$ $\{ \rho_x : x \in \{ 0, 1\}^n \}$ is parity-oblivious. 
\end{lemma} 

\begin{proof} 
Protocols involving shared entanglement and sending one classical bit have limited guessing probabilities for functions such as parity~\cite{PW12}. In particular, it can be shown that the biases $\alpha_S$ of learning $x_S$ satisfy 
\begin{equation*} 
\sum_{S \subseteq \{ 0,1 \}^n \setminus \emptyset} \alpha_S^2 \le 1. \label{eqn:hyperbitbound}
\end{equation*} 
In the $\RAC$ above, we have  
\[ \sum_{S: |S| = 1} \alpha_S^2 \geq n \cdot \left(\frac{1}{\sqrt{n}}\right)^2 = 1 \] 
implying $\alpha_S = 0$ for all $S$ of size $2$ or greater, implying it is parity-oblivious. 
\end{proof} 

This concludes a construction of a quantum $\POMn$ with optimal bias. However, we have not yet proved Theorem~\ref{main} since the encoding dimension is too high. We now discuss a small modification to simultaneously reduce the dimension of the $\RAC$ and to increase its device-independence.   

\subsection{Removing the classical message} 
\label{ssect:dimred} 

Reducing the dimension of the $\POMn$ $\{ \rho_x : x \in \zo^n \}$ is straightforward. First notice that Bob simply takes his measurement outcome and changes it if $a = -1$ to obtain his guess for $x_t$. We can remove the need for this message if Alice simply changes the value of $x$ for which Bob has the encoding. In other words, instead of sending $a$ to Bob, she just switches every bit of $x$ if $a = -1$. Then Bob's guess $b$ is just as accurate in guessing $x'_t$ where $x'_t = x_t$, if $a = +1$, and $x'_t = \overline{x_t}$, if $a = -1$. (Note that this is similar to what was done in Subsection \ref{ssec:GameToProtocol}). Therefore, Bob now has an encoding of $x'$, which we denote by $\{ \sigma_{x'} : x' \in \{ 0, 1 \}^n \}$. Notice that $\sigma_{x'}$ is a state on only $\lfloor n/2 \rfloor$ qubits which coincides with the optimal $\POMn$ for $n=2$ and $n=3$ previously discussed. 

It is easy to see that this new $\RAC$ has bias at least $1/\sqrt{n}$ by construction. We now argue that this quantum $\RAC$ is still parity-oblivious.  
  
\begin{lemma} 
The quantum $\RAC$ $\{ \sigma_{x'} : x \in \{ 0, 1 \}^n \}$ is parity-oblivious.  
\end{lemma} 

\begin{proof} 
We prove that any $S$-parity hidden from Bob in the quantum $\RAC$ $\{ \rho_x : x \in \zo^n \}$ is still hidden from Bob in the quantum $\RAC$ $\{ \sigma_{x'} : {x'} \in \zo \}$. 
Let $\tau_{x,a}$ be Bob's post-measured state immediately after Alice used measurement $A_x$ and received output $a$. We can write 
\[ 
\rho_x := \frac{\tau_{x,0} \otimes \kb{0}}{2}  + \frac{\tau_{x,1} \otimes \kb{1}}{2} 
\quad \text{ and } \quad 
\sigma_{x} := \frac{\tau_{x,0} + \tau_{\bar{x},1}}{2}.  
\] 
Fix a subset $S \subseteq [n]$ of size at least $2$. Since $x_S$ is hidden from Bob in the $\RAC$ $\{ \rho_x : x \in \zo^n \}$, we have $\sum_{x \, : \, x_S = 0} \rho_x = \sum_{x \, : \, x_S = 1} \rho_x$. Thus,  
\[  
\left( \sum_{x \, : \, x_S = 0} \tau_{x,0} \right) \otimes \frac{\kb{0}}{2} + 
\left( \sum_{x \, : \, x_S = 0} \tau_{x,1} \right) \otimes \frac{\kb{1}}{2}  
=
\left( \sum_{x \, : \, x_S = 1}  \tau_{x,0} \right) \otimes \frac{\kb{0}}{2} + 
\left( \sum_{x \, : \, x_S = 1}  \tau_{x,1} \right) \otimes \frac{\kb{1}}{2} 
\]
implying that 
\begin{equation} 
\sum_{x \, : \, x_S = 0} \tau_{x,0} = \sum_{x \, : \, x_S = 1} \tau_{x,0} \quad \text{ and } \quad \sum_{x \, : \, x_S = 0} \tau_{x,1} = \sum_{x \, : \, x_S = 1} \tau_{x,1}. 
\label{eqn:tauequiv}
\end{equation} 
Now we can write 
\begin{eqnarray*}
\sum_{x \, : \, x_S = 0} \sigma_x 
& = & 
\sum_{x \, : \, x_S = 0} \tau_{x,0} 
+  
\sum_{x \, : \, x_S = 0} \tau_{\bar{x},1} 
= 
\sum_{x \, : \, x_S = 1} \tau_{x,0} 
+  
\sum_{x \, : \, x_S = 1} \tau_{\bar{x},1} 
=  
\sum_{x \, : \, x_S = 1} \sigma_x 
\end{eqnarray*} 
using Equation~(\ref{eqn:tauequiv}) implying that the quantum  $\RAC$ $\{ \sigma_x : x \in \zo^n \}$ also hides $x_S$ from Bob. 
\end{proof} 

Since the quantum $\POMn$ $\{ \sigma_x : x \in \zo^n \}$ has optimal bias and uses only $\lfloor n/2 \rfloor$ qubits, 
this concludes the proof of our main result, Theorem~\ref{main}.  
 
\subsection{{Making our optimal $\POMn$ device-independent}} \label{ssect:DI}

We first point out that our preparation of the $\POMn$ $\{ \rho_x : x \in \zo^n \}$ is not secure in the device-independent model for the following reason. Suppose the measurements are not trusted, in the sense that Bob controls them. Consider the case when Alice's measurement, which depends on $x$, simply outputs $a := x_1 \oplus x_2$. Then, when $a$ is sent to Bob, he now knows some information about the parities of $x$. 

Note that in the preparation of the $\RAC$ $\{ \sigma_{x'} : x' \in \{ 0, 1 \}^n \}$ the even parities of $x'$ are hidden from Bob by the no-signalling principle since they are equal to those for $x$ and there is no communication from Alice to Bob in the preparation of $\sigma_{x'}$ after Alice {chooses} $x$. This is the basis for our preparation of a quantum $\RAC$ secure in the device-independent security model. 
 
There is one small caveat however that does not affect the security, but may change how $x'$ is generated. That is, it may not be generated uniformly now. Consider the case when Bob controls the measurement and decides to set $a := x_1$. Then, $\sigma_{x'}$ will never be prepared if $x'_1 = 1$. However, there is an easy fix at the price of adding one classical bit of communication. 

Before sending Bob's part of the (supposed) maximally entangled state, Alice can choose a random bit $d$ and send that to Bob as well. (We assume Alice can flip a random coin without being effected from a malicious Bob.) Then Alice takes the XOR of $d$ with her measurement outcome, and proceeds as usual. Bob can easily adjust his guess for $x'_t$ using the value of $d$. 

It is easy to see that this preparation of the $\POMn$ $\{ \sigma_{x'} : x' \in \zo^n \}$ hides the even parities, in the device-independent model, against any malicious Bob respecting the no-signalling principle, thus proving Theorem~\ref{thm:DI}. 
   
It is not the case {that} the preparation of the $\POMn$ $\{ \sigma_{x'} : x' \in \zo^n \}$ hides the odd parities as well in the device-independent model. For example, suppose Bob controls the measurement such that $a := x_1 \oplus x_2 \oplus x_3$. Then 
\[ x'_1 \oplus x'_2 \oplus x'_3 = (x_1 \oplus a \oplus d) \oplus (x_2 \oplus a \oplus d) \oplus (x_3 \oplus a \oplus d) = d \]
for any choice of $x$. Thus, if Bob controls the measurements, he can make it so that Alice's first three bits always have parity $d$, {of} which he knows the value! 

We leave it as an open problem to see if there exists a  preparation of an optimal $\POMn$ that is still parity-oblivious in the device-independent model against either no-signalling or quantum adversaries. 
  
\section*{Conclusion}   

We have provided the optimal bounds for parity-oblivious quantum $\RAC$s and showed that they are asymptotically better than the optimal classical ones. We discussed how these optimal $\RAC$s provide a large non-contextuality inequality violation and resolve the main open problem in a work of Spekkens, Buzacott, Keehn, Toner, and Pryde~\cite{SBKTP09}. We also studied optimal bounds for a related version of these $\RAC$s which only hide the even parities. We showed their equivalence to a non-local game and explained why even-parity-obliviousness was the correct notion of $\RAC$ in this setting. After constructing a family of optimal parity-oblivious $\RAC$s, we discussed how to make the even parities secure in the device-independent model.

We end with an open question. We have seen how even-parity-obliviousness plays a key role in our analysis,  especially when making $\RAC$s secure in the device-independent setting. This raises the question: What can be said about the optimal bias for \emph{odd}-parity-oblivious $\RAC$s? We have discussed how they can have greater bias than the two variants studied in this work, but perhaps the optimal bias can still be expressed as some nice function of the number of encoded bits. Also, it would be interesting to see if they can be made secure in the device-independent model as well.

\section*{Acknowledgements} 

This research was supported by the French National Research Agency, through CRYQ (ANR-09-JCJC-0067) and by the
European Union through the ERC project QCC. 


\nocite{VW11}
\nocite{CKS14}
\nocite{Tsi87}
\nocite{PW12}
\nocite{SBKTP09}
\nocite{Nay99}
\nocite{CSUU08}
\nocite{CHSH69} 

\bibliographystyle{alpha} 
\bibliography{paper}

\newcommand{\etalchar}[1]{$^{#1}$}
\begin{thebibliography}{BCWdW01}

\bibitem[ANTV99]{ANTV99}
A.~Ambainis, A.~Nayak, A.~{Ta-Shma}, and U.~Vazirani.
\newblock Dense quantum coding and a lower bound for 1-way quantum automata.
\newblock In {\em Proceedings of the 31st Annual ACM Symposium on Theory of
  Computing}, pages 376 -- 383, 1999.

\bibitem[ANTV02]{ANTV02}
A.~Ambainis, A.~Nayak, A.~{Ta-Shma}, and U.~Vazirani.
\newblock Dense quantum coding and quantum finite automata.
\newblock {\em Journal of the ACM}, 49(4):496--511, 2002.

\bibitem[BARdW08]{Ben-AroyaRW08}
A.~Ben-Aroya, O.~Regev, and R.~de~Wolf.
\newblock A hypercontractive inequality for matrix-valued functions with
  applications to quantum computing and {LDC}s.
\newblock In {\em FOCS}, pages 477--486, 2008.

\bibitem[BBBW83]{BBBW83}
C.~Bennett, G.~Brassard, S.~Breidbard, and S.~Wiesner.
\newblock Quantum cryptography, or unforgeable subway tokens.
\newblock In {\em Advances in Cryptology CRYPTO 1982}, pages 267--275, 1983.

\bibitem[BCWdW01]{BCWdW01}
H.~Buhrman, R.~Cleve, J.~Watrous, and R.~de~Wolf.
\newblock Quantum fingerprinting.
\newblock {\em Phys. Rev. Lett.}, 87:167902, Sep 2001.

\bibitem[BJK04]{bjk:q1way}
Z.~{Bar-Yossef}, T.~S. Jayram, and I.~Kerenidis.
\newblock Exponential separation of quantum and classical one-way communication
  complexity.
\newblock In {\em Proceedings of 36th ACM STOC}, pages 128--137, 2004.

\bibitem[BRSdW12]{BRSdW12}
H.~Buhrman, O.~Regev, G.~Scarpa, and R.~de~Wolf.
\newblock Near-optimal and explicit {B}ell inequality violations.
\newblock {\em Theory of Computing}, 8(27):623--645, 2012.

\bibitem[CHSH69]{CHSH69}
J.~Clauser, M.~Horne, A.~Shimony, and R.~Holt.
\newblock Proposed experiment to test local hidden-variable theories.
\newblock {\em Physical Review Letters}, 23(15):880--884, 1969.

\bibitem[CKS14]{CKS14}
A.~Chailloux, I.~Kerenidis, and J.~Sikora.
\newblock Strong connections between quantum encodings, non-locality and
  quantum cryptography.
\newblock {\em Physical Review A}, 89:022334, 2014.

\bibitem[CSUU08]{CSUU08}
R.~Cleve, W.~Slofstra, F.~Unger, and S.~Upadhyay.
\newblock Perfect parallel repetition theorem for quantum {XOR} proof systems.
\newblock {\em Computational Complexity}, 17(2):282--299, 2008.

\bibitem[DV10]{DeV10}
A.~De and T.~Vidick.
\newblock Near-optimal extractors against quantum storage.
\newblock In {\em STOC}, pages 161--170, 2010.

\bibitem[Fra45]{Fra45}
J.~S. Frame.
\newblock Mean deviation of the binomial distribution.
\newblock {\em The American Mathematical Monthly}, 52(7):377--379, 1945.

\bibitem[GIKM98]{GIKM98}
Y.~Gertner, Y.~Ishai, E.~Kushilevitz, and T.~Malkin.
\newblock Protecting data privacy in private information retrieval schemes.
\newblock In {\em JCSS}, pages 151--160. ACM Press, 1998.

\bibitem[GKK{\etalchar{+}}08]{GavinskyKKRW08}
D.~Gavinsky, J.~Kempe, I.~Kerenidis, R.~Raz, and R.~de~Wolf.
\newblock Exponential separation for one-way quantum communication complexity,
  with applications to cryptography.
\newblock {\em SIAM J. Comput.}, 38(5):1695--1708, 2008.

\bibitem[HIN{\etalchar{+}}06]{HINRY}
M.~Hayashi, K.~Iwama, H.~Nishimura, R.~Raymond, and S.~Yamashita.
\newblock (4,1)-quantum random access coding does not exist---one qubit is not
  enough to recover one of four bits.
\newblock {\em New Journal of Physics}, 8(8):129, 2006.

\bibitem[Hol73]{Hol73}
A.~Holevo.
\newblock Some estimates of the information transmitted by quantum
  communication channels.
\newblock {\em Problemy Peredachi Informatsii}, 9:3--11, 1973.

\bibitem[INRY07]{IwamaNRY07}
K.~Iwama, H.~Nishimura, R.~Raymond, and S.~Yamashita.
\newblock Unbounded-error one-way classical and quantum communication
  complexity.
\newblock In {\em ICALP}, pages 110--121, 2007.

\bibitem[KdW04]{KdW04}
I.~Kerenidis and R.~de~Wolf.
\newblock Quantum symmetrically-private information retrieval.
\newblock {\em Information Processing Letters}, 90(3):109--114, 2004.

\bibitem[LPY{\etalchar{+}}12]{PhysRevA.85.052308}
H.-W. Li, M.~Paw\l{}owski, Z.-Q. Yin, G.-C. Guo, and Z.-F. Han.
\newblock Semi-device-independent randomness certification using $n \rightarrow
  1$ quantum random access codes.
\newblock {\em Phys. Rev. A}, 85:052308, May 2012.

\bibitem[Nay99]{Nay99}
A.~Nayak.
\newblock Optimal lower bounds for quantum automata and random access codes.
\newblock {\em Proceedings of 40th IEEE Symposium on Foundations of Computer
  Science}, 0:369--376, 1999.

\bibitem[OW10]{OW10}
J.~Oppenheim and S.~Wehner.
\newblock The uncertainty principle determines the non-locality of quantum
  mechanics.
\newblock {\em Science}, 330:6007:1072--1074, 2010.

\bibitem[PW12]{PW12}
M.~Paw\l{}owski and A.~Winter.
\newblock From qubits to hyperbits.
\newblock {\em Phys. Rev. A}, 85:022331, 2012.

\bibitem[PZ10]{PZ10}
M.~Paw\l{}owski and M.~\.{Z}ukowski.
\newblock Entanglement-assisted random access codes.
\newblock {\em Phys. Rev. A}, 81:042326, Apr 2010.

\bibitem[Raz99]{raz99}
R.~Raz.
\newblock Exponential separation of quantum and classical communication
  complexity.
\newblock In {\em Proc. 31st Annual ACM Symposium on Theory of Computing},
  pages 358--367, New York, NY, USA, 1999. ACM.

\bibitem[RK11]{RegevK11}
O.~Regev and B.~Klartag.
\newblock Quantum one-way communication can be exponentially stronger than
  classical communication.
\newblock In {\em STOC}, pages 31--40, 2011.

\bibitem[SBK{\etalchar{+}}09]{SBKTP09}
R.~W. Spekkens, D.~H. Buzacott, A.~J. Keehn, B.~Toner, and G.~J. Pryde.
\newblock Preparation contextuality powers parity-oblivious multiplexing.
\newblock {\em Physical Review Letters}, 102:010401, 2009.

\bibitem[Tsi87]{Tsi87}
B.~Tsirelson.
\newblock Quantum analogues of the {B}ell inequalities: The case of two
  spatially separated domains.
\newblock {\em Journal of Soviet Mathematics}, 36:557--570, 1987.

\bibitem[VW11]{VW11}
T.~Vidick and S.~Wehner.
\newblock Does ignorance of the whole imply ignorance of the parts?
\newblock {\em Physical Review Letters}, 107:030402, 2011.

\bibitem[Weh08]{W08}
S.~Wehner.
\newblock Unpublished note, 2008.

\end{thebibliography}


\end{document}